\documentclass{article}

\usepackage{arxiv}

\usepackage[T1]{fontenc}    
\usepackage{hyperref}       
\usepackage{url}            
\usepackage{booktabs}       
\usepackage{amsfonts}       
\usepackage{nicefrac}       
\usepackage{microtype}      
\usepackage{amsmath}
\usepackage{amssymb}
\usepackage{dsfont}
\usepackage{amsmath, amsthm}
\usepackage{graphicx}
\usepackage{upgreek}
\usepackage[colorinlistoftodos]{todonotes}
\usepackage{datetime}
\theoremstyle{definition}

\newtheorem{theorem}{Theorem}[section]
\newtheorem{lemma}[theorem]{Lemma}

\date{}
\title{A topological characterization of the existence of  $w$-stable  sets}

\author{
  Athanasios Andrikopoulos\thanks{Professor  (https://www.ceid.upatras.gr/webpages/faculty/aandriko/)} \\
  Dept. of Computer Engineering and Informatics\\
  University of Patras\\
  Patras, 26504, Greece \\
  \texttt{aandriko@ceid.upatras.gr} \\
\And
Nikolaos Sampanis\\
 Dept. of Computer Engineering and Informatics\\
 University of Patras\\
 Patras, 26504, Greece \\
 \texttt{nsampanis@upatras.gr} \\
}

\begin{document}.
\maketitle

\begin{abstract}
The theory of optimal choice sets is a solution theory that has a long and well-established tradition in 
social choice and game theories.
Some of important general solution concepts of choice problems when the set of best alternatives does 
not exist (this problem occurs when the preferences yielded by an economic process are cyclic) is the Stable 
Set (Von Neumann-Morgenstern set)
and its variants (Generalized Stable set, Extended Stable set, 
$m$-Stable set and $w$-Stable set). The theory of $w$-stable sets solution is more realistic because: (1)
It solves the existence problem of solution; (2) It expands the notions of maximal alternative set and (3) 
The concept of stability is defined in such a way as to prevent a chosen alternative from being dominated 
by another alternative and sets this stability within the solution.
In this paper, we present a topological characterization of the existence 
of $w$-Stable sets solution of arbitrary binary relations over non-finite sets of alternatives.
\end{abstract}


\section{Introduction} 
The classical rationality conditions in choice theory formalize the thesis that to choose rationally is to choose in such a way 
that no other choice would have been better, or preferable. That is, each individual makes choices by selecting from each 
feasible set
of alternatives, those which maximize his own preference relation.
According to this hypothesis, the set of choices from a given set of alternatives in which a dominance relation 
is defined consists of the set of maximal elements with respect to this dominance relation, which is known as the {\it core}.
The ordered pair $(X,R)$, where $X$ is a (finite or infinite)
non-empty set of mutually exclusive alternatives 
and $R$ is a dominance relation over $X$ is called {\it abstract decision problem}. 
The set of maximal elements (core) of an abstract decision problem is often 
empty. In this case, it is important to specify criteria that will provide reasonable sets of alternatives as solutions. 
In the choice and game theories, a number of theories, called general solution theories, have been proposed to 
take over the role of maximality 
in the absence of maximal elements. 
Any solution that includes the set of maximal alternatives 
is called {\it core-inclusive}.
Because general solution theories generalize the notion of core, a logical requirement is that they be core-inclusive.
Some of the most important general solution concepts is the Schwartz set
which is equivalent 
to the admissible set by Kalai and Schmeidler or to the dynamic solutions concept of Shenoy \cite{she} in game theory.
The Schwartz set is not only non-empty for every finite abstract decision problem, but also core-inclusive. 
However, this solution has a disadvantage. It may include all the alternatives under consideration. That means that 
this solution may not discriminate alternatives at all.
A different approach to finding general solution concepts for solving choice problems, which is still being studied 
and improved today is the concept of stable set introduced by Von Neumann and Morgenstern \cite{von}.
Stable sets solution is core-inclusive and behaves well in acyclic dominance relations. However, 
the theory of stable sets has a significant flaw in that it can be empty in the case of odd cycles.
To avoid this particular problem, Van Deemen \cite{van} introduced the notion of the generalized stable set which 
in finite sets of alternatives is able to produce a solution for every possible cyclic dominance relation.
Andrikopoulos \cite{and} provides a topological characterization for the existence of the generalized stable set 
in infinite sets of alternatives.
Stable sets solution have an additional weakness, namely,
it is possible that an alternative in a stable set is dominated by an alternative outside this set. Because of this, Peris and
Subiza \cite{per} introduced a reformulation of stable sets called $m$-stable sets. They have shown that $m$-stable sets exist 
for every abstract decision problem, and that alternatives in an $m$-stable set are free from being dominated by any 
alternative outside this set. However, this solution has other disadvantages. First, the solution of the $m$-stable set may 
include all the alternatives under consideration and so that this solution does not discriminate between alternatives at all.
Second, an $m$-stable set $V$ may not be free from inner contradiction in the sense that an alternative $x$ in $V$ may be 
dominated by another alternative $y$ in the same set. This is a violation of what Von Neumann and Morgenstern 
called internal stability. 
Han and Van Deemen in \cite{han} propose an alternative solution called $w$-stable set which 
can accommodate both mentioned 
disadvantages of $m$-stable sets. They prove that $w$-stable sets exist and are proper subsets of the ground set for 
any abstract decision problem. Moreover, no alternative in a $w$-stable set is dominated by another alternative in this set. 
In other words, $w$-stable sets satisfy internal stability. In this respect, the notion of $w$-stable sets differs fundamentally 
from the concept of $m$-stable sets.

 In a characterization of the existence of non-empty choice sets in infinite sets of alternatives in which 
 an arbitrary dominance relation is defined, the usual approach is to assume topological notions such 
 as compactness and continuity. Building on this fact, in this paper we give two characterizations of the existence of
of the $w$-stable sets solution for arbitrary abstract decision problems. 
Beyond the concept of compactness that we use to prove the two basic existence theorems,
the other basic concepts we use are the notion of generalized upper tc-semicontinuity defined by Andrikopoulos 
in \cite{and1} and the notion of contraction relation which is defined in \cite{BCM}.

\section{Notations and definitions}
An abstract decision problem is divided into two parts. One is an arbitrary set $X$ of alternatives (called the {\it ground set}) 
from which an individual or group must select. In most cases, there are at least two alternatives to choose from. Otherwise, 
there is no need to make a decision. The other is a dominance relation over this set, which reflects preferences or evaluations 
for different alternatives.
Preferences or evaluations over $X$ are modelled by a binary relation $R$.
When representing abstract decision problems, the pair $(X,R)$ is used. We sometimes abbreviate $(x,y)\in R$ as $xRy$.
The {\it transitive closure} of $R$ is the relation $\overline{R}$ defined as follows:
For all $x, y\in X$, $(x,y)\in \overline{R}$ if and only if there exist $K\in \mathbb{N}$ and $x_{_0},...,x_{_K}\in X$ such that
$x=x_{_0}, (x_{_{k-1}},x_{_k})\in R$ for all $k\in \{1,...,K\}$ and $x_{_K}=y$.
A subset $D\subseteq X$ is $R$-{\it undominated} if and only if for no $x\in D$ there is a $y\in X\setminus D$ such that $yRx$. 
A subset $Y\subseteq X$ is an $R$-{\it cycle} if for all $x,y\in Y$, we have 
$(x,y)\in \overline{R}$ and $(y,x)\in \overline{R}$. A {\it Top $R$-cycle} is an $R$-cycle which 
which is maximal with respect to set-inclusion.
We say that $R$ is {\it acyclic} if there does not exist an R-cycle. 
An alternative $x\in X$ is $R$-{\it maximal} with respect to a binary relation $R$,
if $(y,x)\in P(R)$ for no $y\in X$.
The traditional choice-theoretic approach takes behavior as rational if there is a binary relation $R$ such that for 
each non-empty subset $A$ of $X$, $\mathcal{C}(A)=\mathcal{M}(A,R)$. To deal with the case where the set 
of maximal elements is empty, Schwartz in \cite[Page 142]{sch} has proposed the general solution concept known as
{\it Generalized Optimal-Choice Axiom} ($\mathcal{G}\mathcal{O}\mathcal{C}\mathcal{H}\mathcal{A}$): 
For each $A\subseteq X$, $\mathcal{C}(A)$ is equivalent to the union of all minimal 
$R$-undominated subsets of $A$. The Schwartz set is the choice set from a given set specified by the $\mathcal{G}\mathcal{O}\mathcal{C}\mathcal{H}\mathcal{A}$
condition.That is, for each $A\subseteq X$, $\mathcal{C}(A)=\displaystyle\bigcup_{D\in\mathcal{D}}$ where
$\mathcal{D}$ is the set of all minimal $R$-undominated subsets of $A$.
Deb in \cite{deb} shows that $\mathcal{C}(A)=\mathcal{M}(A,\overline{P(R)})$ (see also \cite{and3}).
According to \cite[Theorem 19]{and4}  $\mathcal{C}(A)$ is equivalent to the union of all $R$-undominated 
elements and all top $R$-cycles in $X$. It
 is also equivalent to the notion of admissible set in game theory defined by Kalai and Schmeidler in \cite{KS} and 
 the notion of dynamic solutions defined by Shenoy in \cite{she}.
A subset $F$ of $(X,R)$ is called a {\it Von Neumann-Morgenstern stable set}  if (i)
no alternative in $F$ is dominated with respect to $R$ by another alternative in
$F$, and (ii) any alternative outside $F$ is dominated with respect to $R$ by an alternative inside $F$.
The first
property is called {\it internal stability of domination} and the second property {\it external stability of domination}. 
 A subset $F$ of $X$ is called a {\it generalized stable set} if it is a Von Neumann-Morgenstern stable set of $X$ with 
 respect to the transitive closure of $R$.
 A set $F\subseteq X$ is called an $m$-{\it stable set} of $(X,R)$ if (i) $\forall x,y \in F$, if $x\overline{R}y$ then $y\overline{R}x$;
 (ii) $\forall x,y \in F$, there is no $y\in X\setminus F$ such that $y\overline{R}x$.
 A set $F\subseteq X$ is called an $w$-{\it stable set} of $(X,R)$ if
(i) $\forall x,y \in F$, $(x,y)\notin\overline{R}$;
(ii) $\forall x \in F$ and $y\in X\setminus F$, if $y\overline{R}x$ then $x\overline{R}y$.
The following definitions are taken from \cite{van}.
An abstract decision problem $(X,R)$ is called {\it strongly connected} if $x\overline{R}y$ for all $x, y\in X$. 
A {\it strong component} of an abstract decision problem $(X,R)$
is an abstract decision problem $(Y,R|_{_Y})$, $Y\subseteq X$, satisfying the following properties:
($\mathfrak{i}$) $(Y,R|_{_Y})$ is strongly connected;
($\mathfrak{i}\mathfrak{i}$) no abstract decision problem $(Y^{\prime},R|_{_{Y^{\prime}}})$ with $Y^{\prime}\supset Y$ 
is strongly connected.
Note that when an element $x$ is not on any cycle, it forms a singleton strongly connected component $\{x\}$ by itself.
Clearly, the set of strongly connected components form a partition of the space $(X,R)$.
The
{\it contraction} of $(X, R)$ is an abstract decision problem $(\Xi,\widetilde{R})$ where
\par
1. $\Xi=\{X_{_i}\vert i\in I\}$ is the collection of ground sets of the strong components of $(X,R)$;
\par
2. for any $X_{_i}, X_{_j}\in \Xi$, $X_{_i} \widetilde{R} X_{_j}$ if there are $x\in X_{_i}, y\in X_{_j}$ with 
$xRy$. Clearly, $\widetilde{R}$ is acyclic by definition.

In what follows, $\mu(\Xi,\widetilde{R})=\{X^{\ast}_i\vert i\in I\}$ denotes the family of ground sets which 
are $\widetilde{R}$-maximal in $\Xi$.

Let $R$ be a binary relation defined on a topological space $(X,\tau)$.
The relation $R$ is {\it upper} {\it semicontinuous} if for all $x\in X$ the set 
$xR=\{y\in X\vert xRy\}$ is open.
According to 
Alcantud and Rodriguez-Palmero \cite[page 181]{AR} a
binary relation $R$ 
on a topological space $(X,\tau)$
is {\it upper} {\it tc-semicontinuous} if
its transitive closure is upper semicontinuous; i.e., if for all
$x\in X$ the set $x\overline{R}=\{y\in X\vert x\overline{R}y\}$ is open.
The acronym $tc$ refers only to terms based on the initial letters of {\it transitive closure}.
The relation $R$ on $(X,\tau)$ is {\it generalized upper} {\it tc-semicontinuous}
if the set $xP(\overline{R})=\{y\in X\vert xP(\overline{R})y\}$ is open.
Upper semicontinuity implies upper tc-semicontinuity, and both concepts are equivalent for transitive binary relations.
We say that a topological space $(X,\tau)$  is {\it compact} if for 
each collection of open sets which covers $X$ there exists a finite subcollection that also covers $X$.
 
\section{The main result} \vspace{-0.2cm}
We proceed with the main result which is the topological 
 characterization of 
the existence of the $w$-stable set. In what follows, $\mathcal{W}\mathcal{S}(X)$ denotes the family of $w$-stable sets
of an abstract decision problem $(X,R)$.
\vspace{4mm}

\par

\begin{theorem}\label{a1} Let $(X, \tau)$ be a topological space and let $R$ be a binary relation defined in $X$. 
The following conditions are equivalent: 
\par
($\mathfrak{a}$) The $w$-stable set of $R$ in $X$ is non-empty.
\par
($\mathfrak{b}$) There exists a compact topology $\tau$ in $X$ such that $R$ is generalized upper tc-semicontinuous.
\end{theorem}
\begin{proof}($\mathfrak{a}$) $\Rightarrow$ ($\mathfrak{b}$) Suppose that $F\in \mathcal{W}\mathcal{S}(X)$
is non-empty. 
Let $\tau$ be the excluded set topology in $X$ generated by $F$ \cite[p. 48]{SS} (it has as open sets all those 
subsets of $X$ which are disjoint from $F$, together with $X$ itself). Then, $X$ is compact 
under $\tau$ since 
every open cover of $X$ includes $X$ itself. Hence, $\{X\}$ is always a finite subcover. It remains to prove that $R$ is 
generalized tc-upper semicontinuous. 
 In fact, we prove that for each $x\in X$ the sets $\{y\vert xP(\overline{R})y\}$ are open in $\tau$. 
 We have two cases to consider: when $x\in  F$ or not.
 If $x\in  F$ then because of internal stability we have that $\{y\in F\vert xP(\overline{R})y\}=\emptyset$. 
 If $x\notin  F$, then because of external stability we have  
 $\{y\in  F\vert xP(\overline{R})y\}=\emptyset$ as well. Indeed, suppose to the contrary that 
 $(x,y)\in P(\overline{R})$ for some $y\in F$.
 Then, by the external stability, since $(y,x)\notin \overline{R}$ we have that there exists $y^{\prime}\in F$ such that 
 $(y^{\prime},x)\in \overline{R},\ y\neq y^{\prime}$. It follows that $(y^{\prime},y)\in \overline{R}$, a contradiction to internal stability.
 Hence, $R$ is generalized upper tc-semicontinuous.
 \par\smallskip\par\noindent
 ($\mathfrak{b}$) $\Rightarrow$ ($\mathfrak{a}$) 
 Suppose that $\tau$ is compact in $X$ and that $R$ is generalized upper tc-semicontinuous. 
 We prove that the $w$-stable set of $R$ in $X$ is non-empty. To show this, 
 suppose to the contrary that for each $x\in X$ there exists $y\in X$ such that $(y,x)\in \overline{R}$ implies that
 $(x,y)\notin\overline{R}$ or equivalently $(y,x)\in P(\overline{R})$. Therefore, no $\{x\}\in \mathcal{P}(X)$ can be a $w$-stable set (does not satisfy external stability). 
 Therefore,
for every $x\in X$ there exists $y\in X$ such that $yP(\overline{R})x$.
 Thus,
\begin{center}
$X=\displaystyle\bigcup_{y\in X}\{x\in X\vert yP(\overline{R})x\}$
\end{center}
Since the space is compact, there exist $\{y_1,y_2,...,y_n\}$ such that
\begin{center}
$X=\displaystyle\bigcup_{y\in \{y_1,y_2,...,y_n\}}\{x\in X\vert y_iP(\overline{R})x\in X\}$.
\end{center}
Consider the finite set $\{y_1,y_2,...,y_n\}$. Then, for each 
\begin{center}
$x\in X=\displaystyle\bigcup_{y\in \{y_1,y_2,...,y_n\}}\{x\in X\vert y_iP(\overline{R})x\in X\}$
\end{center}
there exists $i=\{1,2,...,n\}$ such that $y_iP(\overline{R})x$.
Since $y_1\in X$, it follows that $y_iP(\overline{R})y_1$ for some $i\in \{1,2,...,n\}$.
If $i=1$, then we have a contradiction. Otherwise, call this element $y_2$.
We have $y_2P(\overline{R}) y_1$.
Similarly, $y_3P(\overline{R})y_2P(\overline{R}) y_1$. As $\{y_1,y_2,...,y_n\}$
is finite, by an induction argument based on this logic, we obtain the existence of a $P(\overline{R})$-cycle 
which is impossible.
 This last 
contradiction shows 
that there exists a
 $x_{0}\in X$ 
 such that for each $y\in X$ we have $(y,x_0)\notin P(\overline{R})$. Then, $(y,x_0)\in \overline{R}$ for some $y\in X$
 implies that
 $(x_0,y)\in\overline{R}$. It follows that $\{x_0\}\subseteq F$ with $F\in \mathcal{W}\mathcal{S}(X)$.
\end{proof}

\begin{lemma}\label{a321}{\rm Let $(X,R)$ be an abstract decision problem, and let $\tau$ be a compact topology in $X$.
Suppose that 
$R$ is upper $tc$-semicontinuous. 
Then, the $w$-stable set is non-empty.
}
\end{lemma}
\begin{proof} Let $x\in X$. If $x$ is an $\overline{R}$-maximal element, then $y\overline{R}x$ implies 
$x\overline{R}y$. Hence, $\{x\}$ is a $w$-stable set.
Otherwise, if $\mathcal{M}({\overline{R}})=\emptyset$ we have that for each $x\in X$ there exists $y\in X$ such that
$(y,x)\in \overline{R}$. Then, as in Theorem \ref{a1} above, we obtain the existence of 
finite set $\{y_1,y_2,...,y_n\}$ which belong to an $R$-cycle $\mathcal{C}_{_\beta}$ such that 
\begin{center}
$X=\displaystyle\bigcup_{y\in \{y_1,y_2,...,y_n\}}\{x\in X\vert y_i\overline{R}x\in X\}$.
\end{center}
By the Lemma of Zorn,\footnote{Zorn's lemma states that every partially ordered set for which 
every chain (that is, every totally ordered subset) has an upper bound contains at least one maximal element.}
the family of all $R$-cycles $(\mathcal{C}_{_\beta})_{_{\beta\in B}}$,
$\mathcal{C}_{_\beta}\subseteq X$ which contain the set $\{y_1,y_2,...,y_n\}$, has a maximal element, 
let $\mathcal{C}_{_{\beta_{_0}}}$. 
Let $x_0\in \mathcal{C}_{_{\beta_{_0}}}$. Then, $\{x_{_0}\}$ is a $w$-stable set. Indded, 
the internal stability of $w$-stable set is evident. 
To prove the external stability of the $w$-stable we suppose that
$y\overline{R}x_{_0}$ holds for some $y\in X$. If 
$y\in\mathcal{C}_{_{\beta_{_0}}}$, then
$x_{_0}\overline{R}y$. 
If $y\overline{R}x_{_0}$ for some $y\in X\setminus \mathcal{C}_{_{\beta_{_0}}}$, 
 then from 
 $(y_{i^{\ast}},y)\in \overline{R}$ for some $i^{\ast}\in \{1,2,...,n\}$ and $(y,y_i)\in \overline{R}$
for each $i\in \{1,2,...,n\}$, we conclude that $y_{i^{\ast}}\overline{R}y\overline{R}y_{i^{\ast}}$ or
$y\in \mathcal{C}_{_{\beta_{_0}}}$ which is impossible. Therefore, $\{x_{_0}\}$ satisfies the external stability of $w$-stable set.
Hence, $\{x_{_0}\}$ is a $w$-stable set.
\end{proof}

\begin{lemma}\label{a221}{\rm Let $(X,R)$ be an abstract decision problem, and let $\tau$ be a compact topology in $X$.
Suppose that 
$R$ is upper $tc$-semicontinuous. 
Then, the family $\mu(\Xi,\widetilde{R})$ of ground sets which 
are $\widetilde{R}$-maximal in $\Xi$ is non-empty.
}
\end{lemma}
\begin{proof} Let $x\in X$.
If $x$ is an $\overline{R}$-maximal element, then $\{x\}$ belongs to a top $R$-cycle (Schwartz set).
Hence, $\{x\}\in \mu(\Xi,\widetilde{R})$.
Otherwise, there exists $y\in X$ such that $y\overline{R}x$.
Similarly if $y$ is an $\overline{R}$-maximal element, then $\{y\}\in \mu(\Xi,\widetilde{R})$.
Otherwise, there exists $y_1\in X$ such that $y_1\overline{R}y\overline{R}x$.
Put
\begin{center}
$A_{_{x}}=\{y\in X\vert \emptyset\subset\overline{R}y\subseteq \overline{R}x\}$.
\end{center}
Since $y_1\overline{R}y\overline{R}x$
we conlcude that $A_{_{x}}\neq \emptyset$.

We now show that $A_{_{x}}$ is closed with respect to $\tau$. 
Suppose that $t$ belongs to the closure of
$ A_{_{x}}$. Then, there exists a net $(t_{_k})_{_{k\in K}}$ in
$A_{_{x}}$ with $t_{_k}\to t$. We have to show that $t\in A_{_{x}}$, i.e., $\overline{R}t\subseteq \overline{R}x$.
Take any $z\in\overline{R}t$. Since $\{w\in X\vert z\overline{R}w\}$ is an open neighborhood of $t$, there exists 
$k^{\prime}\in K$ such that for each $k\geq k^{\prime}$, $z\overline{R}t_{_k}$ holds.
On the other hand, for each $k\geq k^{\prime}$, $t_{_k}\in A_{_{x^{\ast}}}$. Hence, $z\in \overline{R}t_{_k}\subseteq \overline{R}x$.
It follows that $\overline{R}t\subseteq \overline{R}x$ which implies that $t\in A_{_{x}}$.
Therefore, $A_{_{x}}$ is a closed subset of $X$.

If there exists $t^{\ast}\in A_{_{x}}$ which is $\overline{R}$-maximal in $X$, then $t^{\ast}$ belongs to a top $R$-cycle and thus
$\mu(\Xi,\widetilde{R})$ is non-empty. Otherwise,
for each $t\in A_{_{x}}$ there exists $y\in X$ such that
$(y,t)\in \overline{R}$. It follows that $y\in A_{_{x^{\ast}}}$ ($\overline{R}y\subseteq \overline{R}x^{\ast}$). 
Therefore, for each $t\in A_{_{x}}$, the sets $\{y\in X\vert \ y\overline{R}t\}\cap A_{_{x}}$
are open neighbourhoods of $t$ in the relative topology of $A_{_{x}}$, due to upper $tc$-semicontinuity of $R$.
Thus,
the collection $(\{t\in X\vert \ y\overline{R}t\}\cap A_{_{x}})_{_{y\in A_{_{x}}}}$ is an 
open cover of $A_{_{x}}$, that is,
\begin{center}
$A_{_{x}}=\displaystyle\bigcup_{y\in A_{_{x}}}(\{t\in X\vert \ y\overline{R}t\}\cap A_{_{x}})$.
\end{center}
Since $A_{_{x}}$ is compact
there exist $y_{_1},y_{_2},...,y_{_n}\in X$ such that
\begin{center}
$A_{_{x}}=\displaystyle\bigcup_{i=1,2,...,n}
(\{t\in X\vert \ y_{_i}\overline{R}t\}\cap A_{_{x}})$.
\end{center}
We show that among the elements $y_{_1},y_{_2},...,y_{_n}$ there must be an $R$-cycle.
First note that if $i^{\ast}\in \{1,2,...,n\}$, then $y_{_{i^{\ast}}}$ is an element of one of the covering sets
$\{t\in X\vert \ y_{_i}\overline{R}t\}\cap A_{_{x}}$, $i=1,2,...n$.
If $y_{_{i^{\ast}}}\in \{t\in X\vert \ y_{_{i^{\ast}}}\overline{R}t\}\cap A_{_{x}}$, then
we would have an $R$-cycle. Otherwise, for each $i, j\in \{1,2,...,n\}$, $i\neq j$, 
$y_{_i}\in \{t\in X\vert \ y_{_j}\overline{R}t\}\cap A_{_{x}}$. 
Without loss of generality, we assume that
$y_{_1}\in \{t\in X\vert \ y_{_{2}}\overline{R}t\}\cap A_{_{x}}$.
Now, for an arbitrary $i$, we have just the case as we did for
$i=1$, that 
$y_{_i}\in \{t\in X\vert \ y_{_{i+1}}\overline{R}t\}\cap A_{_{x}}$.
Then, $y_{_n}\in \{t\in X\vert \ y_{_{k}}\overline{R}t\}\cap A_{_{x}}$
with $k\in\{1,2,...,n\}$. Thus, we would have an $R$-cycle.

Therefore, there exists an $R$-cycle $\widetilde{\mathcal{C}}$
which contain the elements of a set $M=\{y_{_1},y_{_2},...,y_{_n}\}$.
By the Lemma of Zorn, the family of all $R$-cycles $(\widetilde{\mathcal{C}}_{_\gamma})_{_{\gamma\in \Gamma}}$,
$\widetilde{\mathcal{C}}_{_\gamma}\subseteq 
A_x$, which contain $M$ has a maximal element, which we'll call 
$\widetilde{\mathcal{C}}_{_{\gamma_{_0}}}$.
Therefore, 
by the Lemma of Zorn,
the family of all $R$-cycles $(\mathcal{C}_{_\gamma})_{_{\gamma\in \Gamma}}$,
$\mathcal{C}_{_\gamma}\subseteq A_x$, which contains $M$ has a maximal element, 
let $\mathcal{C}_{_{\gamma_{_0}}}$. Clearly, $\mathcal{C}_{_{\gamma_{_0}}}\in \Xi$.
We prove that $\mathcal{C}_{_{\gamma_{_0}}}\in \mu(\Xi,\widetilde{R})$. 
Indeed, let $X^{\ast} \widetilde{R} \mathcal{C}_{_{\gamma_{_0}}}$ for some $X^{\ast}\in \Xi$.
Then, there are $t\in X^{\ast}, s\in \mathcal{C}_{_{\gamma_{_0}}}$ such that 
$tRs$. If for each $\lambda \in X$ we have $(\lambda,t)\notin \overline{R}$, then
$X^{\ast}=\{t\}$ belongs to the Schwartz set and thus
$X^{\ast}\in\mu(\Xi,\widetilde{R})$.
Otherwise, there exist $\lambda^{\ast}\in X^{\ast}$ such that $(\lambda^{\ast},t)\in\overline{R}$.
Therefore, from $(\lambda,t)\in\overline{R}$, $(t,s)\in R$ and $(s,x)\in\overline{R}$ we conclude that 
$\emptyset\subset\overline{R}t\subseteq \overline{R}x$
which implies that
$t\in A_x$. Since $s\in \mathcal{C}_{_{\gamma_{_0}}}$ we have that $(t,y_i)\in \overline{R}$ for each 
$i\in \{1,2,...,n\}$. On the other hand, since $t\in A_x$ we have $(y_{i^{\ast}},t)\in \overline{R}$ for some $i^{\ast}\in \{1,2,...,n\}$.
Therefore, from $(t,y_{i^{\ast}})\in \overline{R}$ and $(y_{i^{\ast}},t)\in \overline{R}$ we conclude that 
$t\in \mathcal{C}_{_{\gamma_{_0}}}$ which is impossible. Hence, 
$\mathcal{C}_{_{\gamma_{_0}}}\in \mu(\Xi,\widetilde{R})$. Therefore, in any case we have that $\mu(\Xi,\widetilde{R})\neq\emptyset$.
\end{proof}

We now characterize the existence 
of the $w$-stable sets solution for arbitrary binary relations $R$ over non-empty sets of alternatives $X$, via the
contraction relation $(\Xi,\widetilde{R})$ of $(X,R)$.

\begin{theorem}\label{a121}{\rm 
Let $(X,R)$ be an abstract decision problem, and let $\tau$ be a compact topology in $X$.
Suppose that 
$R$ is upper $tc$-semicontinuous.
 Then, 
$W$ is a $w$-stable set of $(X,R)$ if and only if 
$W\subseteq \{x_j\vert x_j\ {\rm exactly\ one\ alternative\ of}\ X^{\ast}_j\in \mu(\Xi,\widetilde{R}), J\subseteq I\}$.}
\end{theorem}
\begin{proof}
Let $(X,R)$ denote an abstract decision problem satisfying the assumptions of the theorem.
Let $W^{\ast}=\{X^{\ast}_j\vert j\in J\}\subseteq \{X^{\ast}_i\vert i\in I\}$ be a subfamily of all ground sets which are $\widetilde{R}$-maximal 
in $\Xi$ and let 
$x_j\in X^{\ast}_j$ for each $j\in J$. This family is non-empty because of Lemma \ref{a221}. 
We prove that $W=\{x_j\vert j\in J\}$ is a 
$w$-stable set. If $(X,R)$ is strongly connected, then $X_i^{\ast}=X$ for all $i\in I$.
In this case, for each $x\in X$, $\{x\}$ is a $w$-stable set of $(X,R)$.
Suppose that
$X^{\ast}_{_{j^{\prime}}}\neq X^{\ast}_{_{j^{\prime\prime}}}$ for at least one pair 
$(j^{\prime},j^{\prime\prime})\in J\times J$.
 We first prove that $W$ satisfies internal stability for $w$-set.
Indeed, let $x_{_{j^{\prime}}}, x_{_{j^{\prime\prime}}}\in W$. We suppose, by way of contradiction, 
that $(x_{_{j^{\prime}}}, x_{_{j^{\prime\prime}}})\in \overline{R}$.
Then, there exists $N\in \mathbb{N}$ and
$x_{_1},x_{_2},...,x_{_N}\in X$
such that
$x_{_{j^{\prime}}}Rx_{_1}Rx_{_2}R...Rx_{_N}Rx_{_{j^{\prime\prime}}}$. Therefore, there are $X_1, X_2,...,X_{_N}\in \Xi$
with $x_{_n}\in X_{_n}$ for all $n\in \{1,2,...,N\}$ satisfying
$X^{\ast}_{_{j^{\prime}}}\widetilde{R} X_1 \widetilde{R} X_2 
\widetilde{R}...\widetilde{R} x_{_N} \widetilde{R} X^{\ast}_{_{j^{\prime\prime}}}$.
Therefore, $X^{\ast}_{_{j^{\prime\prime}}}$ cannot be maximal. This contradiction shows that 
$(x_{_{j^{\prime}}}, x_{_{j^{\prime\prime}}})\notin \overline{R}$.
Similarly we can prove that 
$(x_{_{j^{\prime\prime}}},x_{_{j^{\prime}}})\notin \overline{R}$. Hence, $W$ satisfies internal stability for $w$-set.

To prove that $W$ satisfies external stability for $w$-set, let $x\in W$, $y\in X\setminus W$
such that $(y,x)\in \overline{R}$. We have that $x\in X^{\ast}_{j^{\ast}}$ for some $j^{\ast}\in J$.

We have two cases to consider: when (1) $y\in \displaystyle\bigcup_{i\in I}X^{\ast}_{i}\setminus W$ and when 
(2) $y\in X\setminus \displaystyle\bigcup_{i\in I}X^{\ast}_{i}$.

In case (1), 
the only way for $(y,x)\in \overline{R}$ to hold is for $y\in X^{\ast}_{j^{\ast}}$. But then, since $x\in X^{\ast}_{j^{\ast}}$
we have $(x,y)\in \overline{R}$.

In case (2), since for each $i\in I$, $X_i^{\star}\in \mu(\Xi,\widetilde{R})$ we have that
$(y,x)\notin \overline{R}$ for all 
$y\in X\setminus \displaystyle\bigcup_{i\in I}X^{\ast}_{i}$.
Therefore, $W$ satisfies external stability for $w$-set.

Conversely, let $W$ be a $w$-stable set.
By Lemma \ref{a321} we have that $W\neq \emptyset$.
Let $x_{_0}\in W$ and 
let $\{X^{\ast}_i\vert i\in I\}$ be the set of maximal elements in $(\Xi,\widetilde{R})$.
We prove that $x_{_0}\in X^{\ast}_{\mathfrak{i}}$ for some $\mathfrak{i}\in I$.
If $x_{_0}$ is an $\overline{R}$-maximal element, then $\{x_{_0}\}$ belongs to the Schartz set which implies that
$\{x_{_0}\}\in \mu(\Xi,\widetilde{R})$. Otherwise, there exists $y\in X$ such that $y\overline{R}x_{_0}$.
By the external stability for $w$-set we have that $x_{_0}\overline{R}y$.
By the Lemma of Zorn, 
the family of all $R$-cycles
$(\widehat{\mathcal{C}}_{_\delta})_{_{\delta\in \Delta}}$,
$\widehat{\mathcal{C}}_{_\delta}\subseteq X$, which contain $\{x_{_0},y\}$ 
has a maximal element, let  $\widehat{\mathcal{C}}_{_{\delta_{_0}}}$. We prove that
$\widehat{\mathcal{C}}_{_{\delta_{_0}}}\in \mu(\Xi,\widetilde{R})$.
Indeed, let
$X^{\ast}\in \Xi$ such that $X^{\ast}\widetilde{R} \widehat{\mathcal{C}}_{_{\delta_{_0}}}$.
Then, there exists $t\in X^{\ast}$ and $s\in \widehat{\mathcal{C}}_{_{\delta_{_0}}}$ such that $(t,s)\in R$.
It follows that $(s,x_{_0})\in \overline{R}$ with jointly to $(t,s)\in R$ we conclude that $(t,x_{_0})\in \overline{R}$.
By the external stability for $w$-set we conclude that $(x_{_0},t)\in \overline{R}$ which implies that 
$t\in \widehat{\mathcal{C}}_{_{\delta_{_0}}}$ which is impossible. Therefore, 
$\widehat{\mathcal{C}}_{_{\delta_{_0}}}\in \mu(\Xi,\widetilde{R})$ which completes the proof.
\end{proof}

\par\bigskip\smallskip\par\noindent

\par\noindent
{\it Address}: {\tt {Athanasios Andrikopoulos} \\ {Department of Computer Engineering \& Informatics\\ University of Patras\\ Greece}
\par\noindent
{\it E-mail address}:{\tt aandriko@ceid.upatras.gr}

\end{document}